\documentclass{article}

\usepackage{amsfonts}
\usepackage{mathrsfs}
\usepackage{amssymb}
\usepackage{latexsym}
\usepackage{graphicx}
\usepackage{enumerate}
\usepackage{amsmath}
\usepackage{mathtools}
\usepackage{color}
\usepackage{float}
\usepackage{array}
\usepackage[backend=bibtex]{biblatex}
\mathtoolsset{centercolon}

\usepackage{amsthm}
\newtheorem{theorem}{Theorem}[section]

\newtheorem{proposition}[theorem]{Proposition}

\usepackage{amssymb,tikz}

\newcommand{\defin}[1]{\textbf{#1}}

\newcommand{\lthen}{\rightarrow}

\newcommand{\val}[1]{[\![ #1 ]\!]}
\renewcommand{\phi}{\varphi}

\renewcommand{\L}{\mathcal{L}}

\title{A Note on Proper Relational Structures}
\author{Adam Bjorndahl and Philip Sink}

\begin{document}

\maketitle

\section{Background}

Relational structures are familiar mathematical objects from the world of modal logic; they are the most standard structures with respect to which simple modal languages are semantically interpreted. Let $\L_{n}$ denote the propositional modal language give by
$$\phi ::= p \, | \, \lnot \phi \, | \, \phi \land \psi \, | \, K_{i} \phi,$$
where $p$ is drawn from some countable set of \emph{primitive propositions}, $\textsc{prop}$, and $i \in \{1, \ldots, n\}$. This is a simple language that can be used to reason about $n$ agents and what they know; in this context, $K_{i}\phi$ is read ``agent $i$ knows $\phi$''.

A \defin{relational structure (for $\L_{n}$)} is a tuple $M = (X, (R_{i})_{i=1}^{n}, v)$ where $X$ is a nonempty set, each $R_{i}$ is a binary relation on $X$, and $v: \textsc{prop} \to 2^{X}$ is a valuation. Semantics are standard:
\begin{eqnarray*}
M,x \models p & \textrm{iff} & x \in v(p)\\
M,x \models \lnot \phi & \textrm{iff} & M,x \not\models \phi\\
M,x \models \phi \land \psi & \textrm{iff} & M,x \models \phi \textrm{ and } M,x \models \psi\\
M,x \models K_{i} \phi & \textrm{iff} & R_{i}(x) \subseteq \val{\phi}
\end{eqnarray*}
where $R_{i}(x) = \{y \in X \: : \: xR_{i}y\}$ and $\val{\phi} = \{x \in X \: : \: M,x \models \phi\}$. Typically the $R_{i}$ are assumed to be at least reflexive, encoding factivity of knowledge ($K_{i} \phi \lthen \phi$), and often they are taken to be equivalence relations, encoding positive and negative introspection ($K_{i}\phi \lthen K_{i} K_{i} \phi$ and $\lnot K_{i} \phi \lthen K_{i} \lnot K_{i} \phi$, respectively). 

A relational structure is called \defin{proper} if it contains no pair of distinct points $x,y \in X$ such that for all $i$, $xR_{i}y$. This condition arises in the context of \emph{simplicial semantics} for epistemic logic, in which formulas of $\L_{n}$ are interpreted in simplicial complexes
(see \cite{Death} and \cite{KaSC} for the basic definitions of this framework). Properness is a technical condition that allows for a natural truth-preserving translation from certain relational structures, in particular those where the relations are equivalence relations, to simplicial complexes: the translation proceeds by constructing the nodes of the simplicial complex out of the equivalence classes in the relational model; very roughly speaking, when properness fails, there are not ``enough'' equivalence classes to build a rich enough complex.

Although the properness assumption has been deployed repeatedly to facilitate this translation
\cite{KaSC,DoA,Death,SimpDEL1,FA,SimpDEL2},
to the best of our knowledge there has been no systematic study of just how restrictive it is. In this note we show that in fact properness is not restrictive at all: \textit{every} relational structure is equivalent (via bisimulation) to a proper relational structure.%
\footnote{The only other formal discussion of this topic that we are aware of occurs in \cite{Death}, where the translation of a particular ``canonical'' Kripke model, which is not itself proper, is conjectured to be bisimilar to a proper Kripke model using a broad ``unwinding'' method. Unfortunately, this technique would not, in general, preserve properties like transitivity and symmetry, limiting its usefulness in the context of simplicial semantics, where these properties are typically necessary.}




\section{Result}

We begin with the finite case: let $M = (X, (R_{i})_{i=1}^{n}, v)$ be a relational structure in which $X$ is finite; say $|X| = m$ and $X = \{x_{1}, \ldots, x_{m}\}$. We will construct a new relational structure $\tilde{M} = (\tilde{X}, (\tilde{R}_{i})_{i=1}^{n}, \tilde{v})$, prove that it is proper, and exhibit a surjective, bounded morphism from $\tilde{M}$ to $M$. For a review of bounded morphisms, we direct the reader to \cite{BB}.

To begin, set $\tilde{X} = X \times X$; it is helpful to picture this new state space as consisting of $m$ disjoint copies of the original set $X$, namely,
$$\tilde{X} = X \times \{x_{1}\} \cup \cdots \cup X \times \{x_{m}\},$$
so the point $(x_{j}, x_{k}) \in \tilde{X}$ may be thought of as the $j$th element of the $k$th copy of $X$. Accordingly, we define
$$\tilde{v}(p) = \{(x_{j},x_{k}) \in \tilde{X} \: : \: x_{j} \in v(p)\};$$
in other words, the primitive propositions true at $(x_{j},x_{k})$ in $\tilde{M}$ are precisely those that are true at $x_{j}$ in $M$.

Lastly we must define the relations, and we do so with the idea of preserving this correspondence in truth between $(\tilde{M}, (x_{j}, x_{k}))$ and $(M, x_{j})$ (indeed, we will ultimately show that projection to the first coordinate is the promised bounded morphism). For $i > 1$, define
$$(x_{j},x_{k}) \tilde{R}_{i} (x_{j'},x_{k'}) \textrm{ iff $k = k'$ and $x_{j} R_{i} x_{j'}$}.$$
Notice that if we imposed this same definition for $i=1$, then $\tilde{M}$ would simply be the disjoint union of $m$ copies of $M$. Instead, we will define $\tilde{R}_{1}$ using a different partition of $\tilde{X}$:
let $=_m$ denote equality modulo $m$, and for each $\ell \in \{0, \ldots, m-1\}$, let
$$\tilde{X}_{\ell} = \{(x_{j}, x_{k}) \in \tilde{X} \: : \: k - j =_m \ell\},$$
Thus, for example, $\tilde{X}_{0} = \{(x_{1}, x_{1}), (x_{2}, x_{2}), \ldots, (x_{m}, x_{m})\}$ is just the ``diagonal'' subset, $\tilde{X}_{1} = \{(x_{1}, x_{2}), (x_{2}, x_{3}), \ldots, (x_{m}, x_{1})\}$, and so on. It is then easy to check the following.

\begin{proposition}
The collection $\{\tilde{X}_{\ell} \: : \: 0 \leq \ell \leq m-1\}$ partitions $\tilde{X}$. Moreover, for each $\tilde{X}_{\ell}$, projection to the first component is a bijection between $\tilde{X}_{\ell}$ and $X$.
\end{proposition}
\begin{proof}
	Suppose $(x_{j},x_{k}) \in \tilde{X}_{\ell} \cap \tilde{X}_{\ell'}$
Then $\ell=_mk-j=_m\ell'$. Since $\ell,\ell'\in\{0,\ldots,m-1\}$, we have that $\ell=\ell'$.
This shows that the $\tilde{X}_{\ell}$ are mutually disjoint.
Moreover, given any $(x_{j},x_{k}) \in \tilde{X}$,
clearly $(x_{j},x_{k})\in\tilde{X}_{\ell}$ for $\ell =_m k-j$; thus,
$\{\tilde{X}_{\ell} \: : \: 0 \leq \ell \leq m-1\}$ partitions $\tilde{X}$.
	
	Define $\pi_{1}:\tilde{X} \to X$ by $\pi_{1}(x_{j},x_{k}) = x_{j}$, namely, projection to the first component.
	Let $(x_{j},x_{k}), (x_{j'},x_{k'}) \in \tilde{X}_{\ell}$, and suppose $\pi_{1}(x_{j},x_{k}) = \pi_{1}(x_{j'},x_{k'})$, so $x_{j}=x_{j'}$ (which means that $j = j'$). Then, since $k-j =_m \ell =_m k'-j'$, we have that $k=k'$ and so $x_{k}=x_{k'}$. This shows $\pi_{1}$ is injective when restricted to any $\tilde{X}_{\ell}$. Next, given $x_j \in X$, let $k \in \{1, \ldots, m\}$ be such that $k =_m j+\ell$; then by definition $(x_{j},x_{k}) \in \tilde{X}_{\ell}$ and $\pi_{1}(x_{j},x_{k}) = x_j$, which shows that $\pi_{1}$ is surjective when restricted to any $\tilde{X}_{\ell}$.
\end{proof}

Having established this new way of breaking $\tilde{X}$ into copies of $X$, we define
$$(x_{j},x_{k}) \tilde{R}_{1} (x_{j'},x_{k'}) \textrm{ iff $k - j =_m k' - j'$ and $x_{j} R_{1} x_{j'}$}.$$
So, like the other relations, $\tilde{R}_{1}$ is also $m$ disjoint copies of $R_{1}$, but skewed across a different partition: one copy on each set $\tilde{X}_{\ell}$.

\begin{proposition}
$\tilde{M}$ is proper.
\end{proposition}
\begin{proof}
	Suppose $(x_{j},x_{k}) \tilde{R}_{i} (x_{j'},x_{k'})$ for each $i \in \{1, \ldots, n\}$. Taking $i=1$, we know that $k-j=_mk'-j'$. Taking $i>1$, we know that $k=k'$. So, $j=j'$, and therefore $(x_{j},x_{k})=(x_{j'},x_{k'})$.
\end{proof}

\begin{proposition}
Projection to the first coordinate is a surjective, bounded morphism from $\tilde{M}$ to $M$.
\end{proposition}
\begin{proof}
	We have already seen that $\pi_{1}$ is surjective, and it
	is clear that $\pi_{1}$ also preserves the truth values of atomic formulas. To show $\pi_{1}$ is a bounded morphism,
	we need to establish the ``back'' and ``forth'' conditions for the relations $R_{i}$,
	for which the only interesting case is $i=1$.
	
	First, the ``back'' condition:
	let $(x_{j},x_{k}) \in \tilde{X}$ and $x_{j'}\in X$, and suppose $\pi_{1}(x_{j},x_{k})R_1x_{j'}$. Fix
	$\ell$
	such that $(x_{j},x_{k}) \in \tilde{X}_{\ell}$.
	Choose $k' \in \{1, \ldots, m\}$ such that $k' =_m j' + \ell$, and consider $(x_{j'},x_{k'})$. We have that $k'-j' =_m \ell$, so $(x_{j'},x_{k'}) \in \tilde{X}_{\ell}$. It follows that $(x_{j},x_{k})\tilde{R}_{1}(x_{j'},x_{k'})$ and of course $\pi_{1}(x_{j'},x_{k'}) = x_{j'}$, as desired.
	
	For the ``forth'' condition, suppose that $(x_{j},x_{k})\tilde{R}_{1}(x_{j'},x_{k'})$; then it is immediate from the definition of $\tilde{R}_{1}$ that $x_{j}R_{1}x_{j'}$.
\end{proof}

The construction for countable models is quite analogous, and perhaps even more straightforward since there is no need to appeal to modular arithmetic. Let $M = (X, (R_{i})_{i=1}^{n}, v)$ be a relational structure in which $X$ is countable; thus there exists a bijection $f \colon X \to \mathbb{Z}$. Define $\tilde{M} = (\tilde{X}, (\tilde{R}_{i})_{i=1}^{n}, \tilde{v})$ by setting
\begin{itemize}
\item
$\tilde{X} = X \times X$;
\item
$\tilde{v}(p) = \{(x,y) \in \tilde{X} \: : \: x \in v(p)\}$;
\item
for $i>1$,
$$(x,y) \tilde{R}_{i} (x',y') \textrm{ iff $y=y'$ and $x R_{i} x'$},$$
and
$$(x,y) \tilde{R}_{1} (x',y') \textrm{ iff $f(y) - f(x) = f(y') - f(x')$ and $x R_{1} x'$}.$$
\end{itemize}

Once again, we have: 
\begin{proposition}
$\tilde{M}$ is proper, and projection to the first coordinate is a surjective, bounded morphism from $\tilde{M}$ to $M$.
\end{proposition}
\begin{proof}
For each $\ell \in \mathbb{Z}$, 
define $\tilde{X}_{\ell}  = \{(x,y) \in \tilde{X} \: : \: f(y) - f(x) = \ell\}$. Analogous to
the above, it is easy to see that the collection $\{\tilde{X}_{\ell} \: : \: \ell \in \mathbb{Z}\}$ is a partition of $\tilde{X}$, and more that $\pi_{1} \colon \tilde{X}_{\ell} \to X$ is a bijection for each $\ell$.

The proof of properness is also completely parallel to the finite case:
suppose
$(x,y) \tilde{R}_{i} (x',y')$ for each
$i \in \{1, \ldots, n\}$.
Taking $i=1$, we know that
$f(y) - f(x) = f(y') = f(x')$;
taking $i>1$, we know that
$y = y'$.
It follows that $f(x) = f(x')$, so $x = x'$ and $(x,y) = (x',y')$.

It is clear as before that $\pi_{1} \colon \tilde{X} \to X$ is surjective and preserves the truth values of atomic formulas. 
Thus, to show that it is a bounded morphism, we just need to establish the ``back'' and ``forth'' conditions for the relations $R_{i}$; once again,
the only interesting case is $i=1$.

As above, ``forth'' is immediate from the definition of $\tilde{R}_{1}$. For ``back'', let $(x,y) \in \tilde{X}$ and suppose that $x R_{1} x'$. Let $y' = f^{-1}(f(y) - f(x) +f(x'))$; then $f(y') - f(x') = f(y) - f(x)$, so by definition we have $(x,y) \tilde{R}_{1} (x',y')$ and of course $\pi_{1}(x',y') = x'$, as desired.
\end{proof}

This translation can also be extended to continuum-sized models in the obvious way: namely, by replacing $\mathbb{Z}$ with $\mathbb{R}$ in the proof above (literally every other part of the proof remains the same). This is useful, for example, if one wishes to apply the transformation to a canonical model, which is typically of size continuum.\footnote{For models of other cardinalities, we could again duplicate the proof given above if we had a structure like $\mathbb{Z}$ or $\mathbb{R}$ but with the needed cardinality.
And indeed, such a structure is guaranteed to exist by the L\"owenheim-Skolem theorem.
In fact,
we only really need
to apply the L\"owenheim-Skolem theorem to the group axioms, and not $Th(\mathbb{R})$;
the proof above uses only the existence of additive inverses and the fact that $a-b=a-c$ implies $b=c$.}

Finally, we observe that many properties of the relations $R_{i}$ are preserved by this translation: for example, it is easy to see that if $R_{i}$ is reflexive, symmetric, transitive, serial, or Euclidean, then $\tilde{R}_{i}$ is as well---in all cases the proof follows immediately from the fact that $\tilde{R}_{i}$ can be represented as a disjoint union of ``copies'' of $R_{i}$.
\printbibliography[
heading=bibintoc,
title={Sources}]

\end{document}